\documentclass[leqno]{amsart}

\newcommand{\old}[1]{}

\renewcommand{\emph}[1]{\textit{#1}}
\usepackage{enumitem,amsmath,latexsym,amssymb,amsthm,url}
\usepackage{pstricks,pst-node}
\usepackage{pdftricks}
\begin{psinputs}
  \usepackage{pstricks,pst-node}
\end{psinputs}
\usepackage{color}\usepackage{graphicx}

\definecolor{brown}{cmyk}{0, 0.72, 1, 0.45}
\definecolor{grey}{gray}{0.5}

 \allowdisplaybreaks

\newcounter{rot}

\def\TB{T_{B\&B}}

\newcommand{\ep}{\varepsilon}

\newcommand{\ignore}[1]{}

\newcommand{\M}[1]{M^{(#1)}}

\def\cA{{\mathcal A}}

\def\cX{{\mathcal X}}

\def\cY{{\mathcal Y}}

\def\cP{{\mathcal P}}
\newcommand{\set}[1]{\left\{#1\right\}}

\def\cP{\mathcal{P}}
\def\cQ{\mathcal{Q}}

\newcommand{\proofend}{\hspace*{\fill}\mbox{$\Box$}}

\def\ii_(#1,#2){i_{#1}^{#2}}

\def\cM{\mathcal{M}}
\def\a{\alpha}
\def\b{\beta}
\def\d{\delta}
\def\D{\Delta}
\def\e{\varepsilon}

\def\F{\Phi}
\def\g{\gamma}

\def\z{\zeta}

\def\th{\theta}

\def\l{\lambda}

\def\p{\pi}

\def\r{\rho}

\newcommand{\poly}{\mathrm{poly}}

\newcommand{\sbs}{\subseteq}

\newcommand{\stm}{\setminus}

\newcommand{\brac}[1]{\left( #1 \right)}

\newcommand{\abs}[1]{\left|#1\right|}

\newtheorem{theorem}{Theorem}[section]

\newtheorem{lemma}[theorem]{Lemma}

{
  \theoremstyle{definition}
  \newtheorem{question}[theorem]{Question}

\newtheorem{definition}[theorem]{Definition}

}

\newtheorem{fact}[theorem]{Fact}
\newcounter{thmtemp}

\newcommand{\nospace}[1]{}

\def\path{\operatorname{PATH}}

\newcommand{\beq}[1]{\begin{equation}\label{#1}}
\def\eeq{\end{equation}}

\def\cL{{\mathcal L}}

\newcommand{\dist}{\mathrm{dist}}

\newcommand\xdn{\cX_n}
\newcommand\ydn{\cY_n}

\newcommand{\flr}[1]{\lfloor #1 \rfloor}

\newcommand{\NN}{\textsc{NN}}
\newcommand{\Greedy}{\textsc{Greedy}}
\newcommand{\ninsert}{\textsc{Near-Insert}}

\begin{document}

\newcommand{\PP}{\mathrm{P}}
\newcommand{\NP}{\mathrm{NP}}
\newcommand{\NPO}{\mathrm{NPO}}
\newcommand{\APX}{\mathrm{APX}}

\newcommand{\tf}{{\mathrm{TF}}}
\newcommand{\hk}{{\mathrm{HK}}}
\newcommand{\hf}{{H\mathrm{F}}}
\newcommand{\etf}{{\mathrm{E}2\mathrm{F}}}
\newcommand{\mst}{\mathrm{MST}}
\newcommand{\tsp}{\mathrm{TSP}}
\newcommand{\lb}{\mathrm{LB}}
\newcommand{\mm}{\mathrm{MM}}
\newcommand{\tmm}{\mathrm{2MM}}
\newcommand{\bmst}{\b_\mst}
\newcommand{\bmsT}[1]{\b_{\mst_{#1}}}
\newcommand{\bmstk}{\b_{\mst_k}}
\newcommand{\btsp}{\b_\tsp}
\newcommand{\bmm}{\b_\mm}
\newcommand{\btf}{\b_\tf}
\newcommand{\blb}{\b_\lb}
\newcommand{\bhk}{\b_\hk}
\newcommand{\bhf}{\b_\hf}
\newcommand{\betf}{\b_\etf}
\newcommand{\btfg}{\b_{\tf_g}}
\newcommand{\btF}[1]{\b_{\tf_{#1}}}
\newcommand{\betfg}{\b_{\etf_g}}
\newcommand{\btfG}[1]{\b_{\tf_{#1}}}
\newcommand{\betfG}[1]{\b_{\etf_{#1}}}
\def\whX{{\widehat X}}
\def\whY{{\widehat Y}}
\def\bxi{{\boldsymbol \xi}}
\def\by{{\bf y}}
\def\bd{{\bf d}}

\title{Scalefree hardness of average-case Euclidean TSP approximation}

\author{Alan Frieze}
\address{Department of Mathematical Sciences\\
Carnegie Mellon University\\
Pittsburgh, PA 15213\\
U.S.A.}
\email[Alan Frieze]{alan@random.math.cmu.edu}
\thanks{Research supported in part by NSF grant DMS-1362785.}
\author{Wesley Pegden}
\email[Wesley Pegden]{wes@math.cmu.edu}
\thanks{Research supported in part by NSF grant DMS-1363136.}

\date{\today}

\begin{abstract}
We show that if $\PP\neq \NP$, then a wide class of TSP heuristics fail to approximate the length of the TSP to asymptotic optimality, even for random Euclidean instances.  Previously, this result was not even known for \emph{any} heuristics (greedy, \emph{etc}) used in practice.  As an application, we show that when using a heuristic from this class, a natural class of branch-and-bound algorithms takes exponential time to find an optimal tour (again, even on a random point-set), regardless of the particular branching strategy or lower-bound algorithm used.
\end{abstract}

\maketitle
\section{Introduction}
In this manuscript, we prove that if $\PP\neq \NP$, then \emph{scalefree} heuristics cannot find asymptotically optimal approximate solutions even to random instances of the Euclidean TSP.  Roughly speaking, \emph{scalefree} heuristics are those which do not work especially hard at small scales.  This has two important consequences.  

First, it shows rigorously for the first time that several simple heuristics used for the TSP (Nearest Neighbor, Nearest-Insertion, \emph{etc.}) cannot approximate the TSP to asymptotic optimality (even in average-case analysis), since these heuristics are all scalefree.  In particular, our result can be seen as a defense of the intricacy of the celebrated polynomial-time approximation schemes of Arora \cite{Ar} and Mitchell \cite{Mit} for the Euclidean TSP, as we can show that the simpler algorithms cannot match their performance on sufficiently large random instances.

The second consequence is a new view on the complexity of the Euclidean TSP versus other ``actually easy'' problems on Euclidean point-sets.  Recall that for problems such as determining whether a Minimum Spanning Tree of cost $\leq 1$ exists, or finding the shortest path between 2 points from a restricted set of allowable edges, the complexity status in the Euclidean case is unknown, as no algorithm is known to efficiently compare sums of radicals.  In particular, it is conceivable that these problems are $\NP$-hard (even if $\PP\neq \NP$) as is the Euclidean TSP \cite{PapaComplete, Gcomplete}.  Just as there are (sophisticated) efficient approximation algorithms for the Euclidean TSP, there are (trivial) efficient approximation algorithms for, say, the MST also: simply carry out Kruskal's algorithm and calculate the length of the spanning tree to some suitable precision, to obtain a good approximation.  Our result allows a rigorous distinction between these types of approximations:  Kruskal's algorithm is scalefree in the sense of our paper, and we show that no algorithm with this property can well approximate the Euclidean TSP.   

In particular, a surprising message from our result is that it is possible to connect Turing machine complexity to the practical difficulty of a problem for Euclidean point-sets, in spite of the unresolved state of the difficulty of comparing sums of radicals.  In particular our result connects traditional worst-case deterministic analysis to \emph{average-case}, \emph{approximate} analysis of the Euclidean TSP, for a certain class of algorithms.

\bigskip

To motivate our definition of scalefree, we recall a simplification of the dissection algorithm of Karp, which succeeds at efficiently approximating the length of the shortest TSP tour in a random point set to asymptotic optimality.
\begin{enumerate}[label=(\arabic*)]
\item Let $s(n)=\flr{\tfrac{n^{1/d}}{\log^{1/d}(n)}}^d$ and divide the hypercube $[0,1]^d$ into $s(n)$ congruent subcubes $Q_1,\dots,Q_{s(n)}$, letting $X_i=X\cap Q_i$ for each $i=1,2,\dots,s(n)$.
\item Using the dynamic programming algorithm of Held and Karp \cite{HK0}, find an optimal tour $T_i$ for each set $X_i$.
\item Patch the tours $T_i$ into a tour through all of $X$.
\end{enumerate}
This algorithm runs in expected time $O(n^2\log n)$ and finds an asymptotically optimal tour; i.e., the length of $T$ is $(1+o(1))$ times the optimal length, w.h.p.  In some sense, its defining feature is that it works hard (running in exponential time in small sets of vertices) on a small scale, and is more careless on a large scale.  Our definition of scalefree is intended to capture algorithms which do not exhibit this kind of behavior.

We will define a \emph{heuristic} $H$ to be a function which takes as input the distance matrix for a point-set $X$, and outputs a structure (in this paper, either a TSP of $X$ or a list of paths through $X$) which depends only on comparisons of sums of distances in $X$.

Formally, given the matrix of distances $d_1,\dots,d_{\binom n 2}$ among points in $X$, for each choice of coefficients $\xi_i\in \{0,1,-1\}$ $(1\leq i\leq \binom n 2)$, the sum $\sum \xi_i d_i$ is either negative, zero, or positive.  If all $3^{\binom{n}{2}}-1$ nontrivial sums are nonzero, we say the points of $X$ are in general position.   Now, $H$ is a heuristic if for any points in general position, the signs of these $3^{\binom{n}{2}}-1$ sums determine the output of $H$.  We call the heuristic $H$ polynomial-time if it is always possible to determine the output of the function by querying the sign of only polynomially-many sums of the form $\sum \xi_i d_i$.




Our definition of scalefree requires that the small-scale behavior of a heuristic can be efficiently simulated (usually, by simply running the heuristic on the local data).  Care must be taken when making this precise, and our precise 
definitions of scalefree, together with proofs that commonly used heuristics satisfy the definition, are given in Section \ref{s.scalefreeness}.  Our main theorem is as follows:



\begin{theorem}
\label{t.length}
If $H(X)$ is a scalefree polynomial-time TSP heuristic and $\PP\neq \NP$, then there is an $\e_H>0$ so that w.h.p, $H(\bar \xdn)$ has length greater than $(1+\e_H)$ times optimal, where $\bar\xdn$ is a discretization of the random set $\xdn\sbs [0,1]^d$.
\end{theorem}
Note that it is possible that a polynomial-time heuristic as we have defined it may not actually be implementable in polynomial time on a Turing machine, since the calls to the comparison blackbox may not be efficiently implementable, depending on the hardness of comparing sums of radicals.  But our theorem applies even just if the number of comparisons is polynomially bounded.



We will also consider the effectiveness of a scalefree heuristic in the context of a Breadth-First Branch and Bound algorithm.  Defined precisely in Section \ref{BB}, this is a branch-and-bound algorithm which branches as a binary tree, which is explored in a breadth-first manner.  Branch-and-bound algorithms used in practice do typically satisfy this definition, as they make use of an LP-based lower bound on the TSP, and branch on the binary $\{0,1\}$-possibilities for fractional variables.
\begin{theorem}\label{thmbb}
Any Breadth-First Branch and Bound algorithm that employs a scalefree heuristic to generate upper bounds w.h.p. requires $e^{\Omega(n)}$ time to complete.
\end{theorem}
We close this section by noting that our proof of Theorem \ref{t.length} provides a recipe to attempt to eliminate the $\PP\neq \NP$ assumption for any specific, fixed scalefree heuristic $H$ as follows.  Let $\cL$ be any decision problem in $\NP$.
\begin{theorem}
\label{t.halg}
For any scalefree polynomial-time TSP heuristic $H(X)$, there is a polynomial-time algorithm $A_H$ such that for any $L\in \cL$, we have:
\begin{enumerate}[label=(\alph*)]
\item If $L\notin \cL$, then $A_H(L)$ returns false.
\item If $L\in \cL$, then $A_H(L)$ returns either true, or exhibits an $\e_H>0$ and a proof that $H(\bar \xdn)$ has length greater than $(1+\e_H)$ times optimal w.h.p.
\end{enumerate}
\end{theorem}

\section{Scalefreeness of heuristics}
\label{s.scalefreeness}
In this section, we will in fact give two definitions of scalefree, the first simpler but more restrictive, and the second more general, but more complicated.

For simplicity of notation, we will let $t=n^{\frac 1 d}$ and work with $\ydn\sbs [0,t]^d$, defined as the rescaling $t\cdot \xdn$.  Observe that with this choice of rescaling, a typical vertex in $\ydn$ is at distance $\Theta(1)$ away from its nearest neighbor.

Our notion of scalefreeness captures a common property of many simple TSP heuristics: namely, the small-scale behavior of the algorithm can be efficiently simulated (because it is essentially governed by the same rules as the large-scale behavior).  To give a precise definition of scalefree,  
we will use the notion of a polynomial-time path-finding heuristic $A_H$ for the Heuristic $H$:

\begin{definition} A polynomial-time path-finding heuristic is an algorithm which has access to a blackbox for making comparisons of sums of distances among the input points, and outputs a list of Hamilton paths through a set of points in general position in time polynomial in the number of points, given the distance matrix through the point set.
\end{definition}

Now we are ready for the simpler of our two definitions of scalefree:
\begin{definition}Call \label{d.firstscalefree}$H$ \emph{scalefree} if there exists a constant $R$ and a polynomial-time path-finding heuristic $A_H$, such that given an input set $X$ for $H$, the following implication holds for all sets $S_p=B(p,1)\cap X$:

\textbf{If:} 
\vspace{-2.5ex}
\begin{enumerate}[label=(\alph*)]
\item
  \label{p.protectedA}
  The annulus $B(p,R)\setminus B(p,1)$ around $S_p$ contains no points of $X$,
\item \label{p.generalA}
  The points of $S_p$ are in general position,
\end{enumerate}

\textbf{Then:} 
$H$ traverses $S$ in a path, which belongs to the list produced produced by $A_H$ when a congruent copy of $S$ is used as the input to $A_H$.
\end{definition}
(Here, $B(p,r)$ is the Euclidean ball of radius $r$ centered at $p$.)

Perhaps the best motivation for this definition is simply the proofs that it is satisfied on some important examples.  We begin with the \textsc{Greedy} heuristic, which produces a tour by adding, at each step, the shortest edge which would not create a non-Hamilton cycle or a vertex of degree 3.  (Note that if there is a tie, the points are not in general position, and we allow \textsc{Greedy} to have undefined behavior).

\begin{fact}
  The Greedy Heuristic $\textsc{Greedy}(X)$ Heuristic is scalefree.
\end{fact}
\begin{proof}
  For $\Greedy(X)$, we let $R=3$, say.  $A_\textsc{Greedy}$ is defined simply to be the Greedy heuristic for choosing a Hamilton path, which adds at each stage the shortest edge which would not create a cycle or a vertex of degree 3; as usual, when the input points are not in general position, the behavior of the $A_\textsc{Greedy}$ can be arbitrary.  

  Given a set $S_p=B(p,1)\cap X$, the distance between any two points in $S_p$ is smaller than the distance across the annulus $B(p,3)\setminus B(p,1)$, and so at some stage of $\textsc{Greedy}$ on $X$, $S_p$ will be covered by a path, while no edges cross the annulus.  Assuming the points of $S_p$ are in general position, this is the same as the path that will be returned by $A_H$.
\end{proof}

This case was particularly simple. For example, $A_{\textsc{Greedy}}$ actually always outputs just a single path, instead of a list.  For the Nearest Neighbor algorithm, things are just slightly more complicated:

\begin{fact}
  The Nearest Neighbor heuristic $\textsc{NN}(X)$ is scalefree.
\end{fact}
\begin{proof}
  Recall that the Nearest Neighbor heurstic $\NN(X)$ begins from a distinguished point $x_0\in X$, and then grows a path by choosing, at each step, the nearest vertex to the current one not already on the path.  (In the last step, the endpoints of the Hamilton path are joined to create a tour.)  To show that $\NN(X)$ is scalefree, we again let $R=3$; we then consider a set $S_p=B(p,1)\cap X$ such that the annulus $B(p,3)\setminus B(p,1)$ contains no points of $X$.

  Observe that as $\NN(X)$ progresses, there will be some first step when the heuristic chooses a point from $S_p$.  Thereafter, since all points in $S_p$ are closer to each other than to points in $X\setminus S_p$, it will exhaust the points in $S_p$ before revisiting $X\setminus S_p$.

  Thus, we let $A_\NN(S)$ be the algorithm which uses the Nearest Neighbor heuristic to choose a path through $S$ \emph{for each choice of the initial vertex $x$} in $S$.  In particular, $A_\NN(S)$ will output up to $|S|$ distinct paths through $S$.  (Again, if any ties would be encountered, the points of $S$ are not in general position and so $A_\NN(S)$ is allowed to have undefined behavior.)    With this choice of $A_\NN(S)$, the path taken by $\NN(X)$ through $S$ will be among the list produced by $A_\NN(S)$, assuming the points of $S$ are in general position.
\end{proof}

These examples show the essential character of the notion of a scalefree heuristic.  But if we restrict ourselves to Definition \ref{d.firstscalefree}, it would seem we cannot hope to show that some other common heuristics are scalefree.  Consider, for example, the Nearest-Insertion heuristic $\ninsert(X)$, which begins with a triangle on vertices $x_1,x_2,x_3\in X$, and then, at each step, grows the existing cycle by finding the vertex nearest to the vertex set of the existing cycle, and inserts the vertex into the existing cycle at minimum cost.  If we consider an isolated collection of vertices $S_p$, predicting the local behavior of $\ninsert(X)$ at $S_p$ seems difficult.  Consider for example the first step at which the insertion of a vertex $s\in S_p$ occurs, joined to vertices $x$ and $y$ already on the cycle.  The next vertex in $S_p$ to be joined to the cycle depends not only on $S_p$ and the vertex $s$ but also potentially on the positions of the vertices $x$ and $y$, which are outside the ``local configuration''.

But it will turn out that this is not really a problem.  The theorems we prove essentially will work with any definition of scalefreeness where the implication Definition \ref{d.firstscalefree} is only required to hold for some reasonable fraction of sets $S_p$.  Since $X$ is itself a random set in our Theorems, this means that we can impose any number of reasonable restrictions on the set $S_p$ in the implication, and still have the resulting notion of scalefree be strong enough for our techniques to give Theorem \ref{t.length}.

One way we will do this is by restricting the implication to sets $S_p$ surrounded by special configurations of points.  To this end, given $X\sbs [0,t]^d$ and $Y=\{y_1,y_2,\dots,y_K\}\sbs [0,2R]^d$, we say $X$ is $(R,\e)$-\emph{protected} by $Y$ at $p$ if 
\[
X\cap (B(p,R)\stm B(p,1))= Y'\approx_\e (Y+p)\quad\text{and}\quad Y'\sbs B(p,\sqrt R).
\]
Here, $A\approx_\e B$ means that there is a bijection $f$ from $A$ to $B$ such that for all $a\in A$, $\dist(a,f(a))<\e$. 

Thus, roughly speaking, a set is $(R,\e)$-protected by $Y$ at $p$ if it is surrounded by a nearly-congruent copy of $Y$ in an annulus containing no other points of $X$.   Notice that requirement \ref{p.protectedA} on $S_p$ in Definition \ref{d.firstscalefree} is simply the requirement that $S_p$ is $(R,\e)$-protected by $Y=\varnothing$. One way we will generalize Definition \ref{d.firstscalefree} is simply by allowing the restriction of the implication to subsets $S_p$ of $X$ which are $(R,\ep)$-protected by a fixed set $Y_H$ other than $\varnothing$.


In particular, the following generalization of scalefreeness adds several restrictions to the sets $S_p$ required to satisfy the implication (thus generalizing the definition) in a way which is tailor-made to be sufficient to easily include common insertion heuristics.

  \begin{definition}\label{d.scalefree}$H$ is \emph{scalefree} if there exists $R,\e$,
    some \emph{gadget} $$Y=\{y_1,y_2,\dots,y_s\}\sbs B(0,\sqrt R)\stm B(0,1),$$ and a polynomial path-finding heuristic $A_H$, such that given an input set $X$ for $H$, there is a bounded-size exceptional set $X_0$ such that the following implication holds

\textbf{If:} We have: 

\vspace{-2.5ex}

\begin{enumerate}[label=(\alph*)]
\item  \label{p.protected} $X$ is $(R,\e)$-protected by $Y$ at $p$.
\item \label{p.general} The points of $B(p,R)\cap X$ are in general position.
\item \label{p.exceptional} $B(p,R)\cap X$ contains no points of the exceptional set $X_0$.
\item \label{p.atapath} The tour $T_H$ found by $H$ traverses $S_p\cup Y$ in a single path $P_H$, whose length is within $\ep$ of the length of the shortest path through $S_p$ with the same endpoints.
\item \label{p.angle} The vertices $x,y\in X\stm P_H$ adjacent in $T_H$ to the endpoints of $P_H$ satisfy $\angle xpq,\angle ypq<\e$, where $q=p+(1,0)$.
\end{enumerate}
\textbf{Then:}  The path $P_H$ by which $H$ traverses $S_p$ belongs to the list produced by $A_H$ when a congruent copy of $S_p$ is used as the input to $A_H$.
  \end{definition}


Again we motivate the applicability of this definition by example.  Recall that the \emph{Nearest Insertion} heuristic $\ninsert(X)$ begins, say, with $T$ as the triangle on distinguished vertices $x_1,x_2,x_3\in X$.  (For Definition \ref{d.scalefree}, we will choose this triple of vertices as the exceptional set $X_0$.)   At each step of the algorithm until $T$ is a tour, the Nearest Insertion algorithm finds the vertex $z$ in $X\stm T$ which is closest to $V(T)$, finds the edge $\{x,y\}\in T$ for which $C=d(x,z)+d(y,z)-d(x,y)$ is minimized, and patches the vertex $z$ in between $x$ and $y$ in the tour, at cost $C$.  

\begin{fact}$\ninsert(X)$ is scalefree.
\end{fact}
\begin{proof}
We let $X_0$ be the set of distinguished vertices $x_1,x_2,x_3\in X$ which are the vertex-set of the initial tour for $\ninsert(X)$.  We will assume $R$ is a large constant and $\ep$ is a small constant, without determining the weakest possible requirements on their magnitudes.
  
To prove scalefreeness, we first observe that given an $S_p$ which is $(R,\ep)$-protected by the gadget $Y$, we need only show the implication of Definition \ref{d.scalefree} holds for $S_p$ assuming that there is only ever one insertion of a vertex from $Y\cup S$ at an edge whose endpoints both lie outside of $Y\cup S$; if more than one such insertion occurs in the running of $\ninsert$, then the final tour chosen by the Heuristic will not intersect $Y\cup S$ at a path, violating condition \ref{p.atapath} in the definition.


We use the gadget $Y$ shown in Figure \ref{f.nnG}.  This consists of, say, 18 equally spaced points at angles $(2k-1)\p/18$ to the horizontal $(k=1,\dots,18)$, on the circle of radius $\sqrt R$ with center $p+(0,\tfrac{\sqrt R}{2})$.

Hypothesis \ref{p.angle} from the implication in Definition \ref{d.scalefree} ensures that the two vertices in $T_H$ adjacent to vertices in $B(R,p)$ are significantly closer to the points $x_1,x_{18}$ than to any other points of $B(R,p)$.   Now the path drawn in Figure \ref{f.nnG} transits $S\cup Y$ optimally given that it uses these endpoints (assuming the route through $S$ is optimal), and in particular, it follows that for a sufficiently small choice of $\ep$, if the hypothesis \ref{p.atapath} and \ref{p.angle} are satisfied in Definition \ref{d.scalefree}, then the tour $T_\ninsert$ traverses $X\cap B(p,R)$ as shown in Figure \ref{f.nnG} (and transits $S$ within $\ep$ of optimally). 

Assuming this is the case, we aim to predict the precise path taken by the heuristic (in particular, in $S$). Our choice of $Y$ ensures that the subtour constructed by $\ninsert(X)$ will contain all of $Y$ before it contains any vertex from $S$; thus, we are guaranteed that the two edges leaving $B(p,R)$ are never used for insertions of points in $S_p$ (as there is always a cheaper insertion using closer edges).  Since no edges with both endpoints outside of $B(p,R)$ are used for insertions after the first insertion, all insertions of points in $S$ have both endpoints in $B(p,R)\cap X$.  In particular, we can use the nearest insertion algorithm locally to determine the resulting path: our path finding Heuristic $A_{\ninsert}$ for this case simply begins with a path from $x_1$ to $x_{18}$ through $Y$, and uses Nearest Insertion to extend this to a path through all of $B(p,R)$.  (Recall that in the case of any ties, the input set is not in general position and so we require nothing of the behavior of $A_{\ninsert}$.)

We note that with small modifications, it is not hard to extend the scalefreeness proof to the \emph{Farthest Insertion} heuristic, which inserts at minimum cost the farthest vertex from the vertex-set of the current subtour.  The main problem is just that the heuristic will visit $S$ before exhausting $Y$.  However, it will only visit one vertex of $S$ before exhausting $Y$, which means that we can (at polynomial cost) simply guess the first vertex of $S$ visited by the heuristic, and so $A_H$ will output a polynomially-long list of candidate paths.
\end{proof}

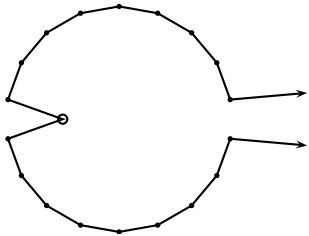
\begin{figure}[t]
\begin{center}
\begin{pdfpic}
\psset{unit=.5cm,dotsize=2pt}
\begin{pspicture}(-5,-5)(10,5)

\SpecialCoor
\dotnodes*(3;10){A}(3;30){B}(3;50){C}(3;70){D}(3;90){E}(3;110){F}(3;130){G}(3;150){H}(3;170){I}(3;190){J}(3;210){K}(3;230){L}(3;250){M}(3;270){N}(3;290){O}(3;310){P}(3;330){Q}(3;350){R}

\rput[tr](A){\small $x_1$}
\rput[br](I){\small $x_9$}
\rput[tr](J){\small $x_{10}$}
\rput[br](R){\small $x_{18}$}


\pnode(-1.5,0){point}
\pscircle(point){.15}

\psline[arrows=<->](5,.7)(A)(B)(C)(D)(E)(F)(G)(H)(I)(point)(J)(K)(L)(M)(N)(O)(P)(Q)(R)(5,-.7)
\psdot(10,-1.15)
\psdot(10,1.15)
\rput(10,-1.5){$y$}
\rput(10,1.5){$x$}

\end{pspicture}
\end{pdfpic}
\end{center}
\caption{\label{f.nnG} (Proving that the Nearest Insertion heuristic is scalefree.) The gadget $Y$ (consisting of the 18 points $x_i$) prevents the optimal tour from entering directly to a vertex in $S\sbs B(p,1)$. The boundary of $B(p,1)$, which contains all points in the set $S_p$, is drawn as the tiny circle.}
\end{figure}

\section{Proof outline}
We begin by giving a broad outline of the proof of Theorem \ref{t.length}, so that the reader can have a preview of the overall structure of the argument.  We begin by discussing the task of proving that a Heuristic satisfies the simpler and more powerful Definition \ref{d.firstscalefree} definition of scalefree.

The proof begins by leveraging Papadimitriou's reduction \cite{PapaComplete} of the $\NP$-complete Set Partition problem to the Hamilton path problem.  Given any instance of the Set Partition problem, Papadimitriou constructs a set of points and a threshold such that a tour shorter than the threshold exists if and only if the Set-Partition problem is feasible.

Next, we surround this point set with a suitable arrangement of two points to create a configuration $\cQ$, which has the property that \emph{if} the optimal tour through a set $X\supset\cQ$ transits $\cQ$ in a single pass, then in doing so, it transits the Papadimitriou set optimally (thus solving the corresponding Set Partition problem).


Next we use the path-finding Heuristic $A_H$ guaranteed to exist by Definition \ref{d.firstscalefree} to define an algorithm to solve the Set Partition problem as follows:
\begin{itemize}
\item Generate the set $\cQ$ as above corresponding to the given Set Partition instance through Papadimitriou's reduction,
\item Use $A_H$ to produce a list of paths through $\cQ$,
\item If any path is below the threshold given by Papadimitriou's reduction, return TRUE; otherwise, return FALSE.
\end{itemize}
We will show the this algorithm can be suitably adapted to the Turing machine setting with polynomial running time, despite the the obvious questions about how to deal with precision issues.

Now, if $\PP\neq \NP$, if must be the case that on some instance of the Set Partition problem, this algorithm gives the wrong answer.  Note that when it answers incorrectly, it necessarily answers FALSE.  (It can fail to find a short tour, but not incorrectly report the existence of a short tour.)

But for this instance of the Set Partition problem, the corresponding set $\cQ$ has the property that any time $\cQ\subset X$, and the Heuristic $H$ finds a tour through $X$, it will transit $\cQ$ suboptimally.

Finally, we will apply a Lemma we proved in \cite{TSPconstants}, which shows that if $X$ consists of $n$ independent uniformly random points in the unit square, than $X$ contains linearly many approximate copies of the set $\cQ$.  These approximate copies will be close enough to $\cQ$ that the tour will still transit them suboptimally, gaining some excess length for each copy, and linearly many such copies are sufficient to ensure a multiplicative $(1+\ep)$ error in the final tour length.

\bigskip

The proofs below are written for the weaker, more complicated Definition \ref{d.scalefree}.  The only important modifications to the above outline are surrounding $\cQ$ with more configurations of points, so that:
\begin{itemize}
\item the set $Y$ required by Definition \ref{d.scalefree} is present,
\item it is necessarily the case that $\cQ$ is transited in a single pass,
\item the angle of entry/exit from $Y$ is controlled, as required by Definition \ref{d.scalefree}.
\end{itemize}

\section{Asymptotic length of scalefree heuristics}\label{alsfh}

Our proof involves a multi-layered geometric construction of a configuration, which, when it exists, will ensure that it contains a special set $S$ satisfying hypotheses \ref{p.protected}, \ref{p.general}, \ref{p.exceptional}, \ref{p.atapath},  \ref{p.angle} of the implication in definition of scalefreeness.  We consider the layers one at a time.

\subsection{Papadimitriou's set $\cP$}
Our proof begins with Papadimitriou's reduction \cite{PapaComplete} to the TSP path problem from the $\NP$-complete Set Partition problem.  This will form the basis for the sets $S$ we wish to apply the definition of scalefreeness to.

Recall that an instance of the Set Partition problem is a family of subsets $A_i,i=1,2,\ldots,M$ of $[N]=\{1,2,\dots,N\}$; the decision problem is to determine whether there is a subfamily which covers $[N]$ and consists of pairwise disjoint sets.

In particular, Papadimitriou shows that for any instance $(\cA,N)$ of the set cover problem, there is a $k$ (polynomial in the size of the Set Partition problem), and a set of $k$ points $\cP=\cP(\cA,N)$ in $[0,\sqrt k]^2$ with distinguished vertices $p$ and $q$ which can be produced in polynomial time, such that for some absolute constant $\e_0$ we have that for any approximator $f_{\e_0}:\cP\to [0,\sqrt k]^2$ with $\dist(f(x),x)<\e_0$ for all $x\in \cP$ that
\begin{enumerate}[label=P\arabic*]
\item \label{Papa.ends}The shortest  TSP path on $\cP$ begins at $f_{\e_0}(p)$ and ends at $f_{\e_0}(q)$.
\item \label{Papa.thres} There is a real number $L$ such that the length of the shortest TSP path on $f_{\e_0}(\cP)$ is either less than $L$ or greater than $L+\e_0$, according to whether the Set Partition instance problem should be answered Yes or No, respectively.
\end{enumerate}
Papadimitriou's discussion does not reference an approximating function like $f_{\e_0}$; the role of this function here is to capture the imprecision which can be tolerated by the construction, which is discussed on page 241 of his paper (where one finds, for example, that we can take, e.g., $\e_0=\tfrac{\sqrt{a^2+1}-a}{100(4a^2+2a)}$ for $a=20$).

\subsection{The set $\cQ$}
Papadimitriou's construction in \cite{PapaComplete} does not have the following property, but it is easy to ensure by simple modification of his construction (using ``1-chains'' to relocate the original $p$ and $q$ to suitable locations):
\begin{enumerate}[resume, label=P\arabic*]
\item \label{P.lines} There is a rhombus $R$ with vertices $p,q$ so that all points in $f_{\e_0}(\cP\stm \{p,q\})$ lie inside $R$ and at least $\e_0$ from the boundary of $R$.
\end{enumerate}
 
When scaled to $[0,\sqrt k]^2$ as we have done here, the minimum TSP path length through Papadimitriou's set will always be less than $C_0 k$ for some absolute constant $C_0>1$ (indeed, this is true even for a worst-case placement of $k$ points in $[0,\sqrt k]^2$ \cite{Few}).  Thus, given the configuration $\cP=\cP(\cA,N)$ with $k=k(\cA,N)$ points and a small $\l>0$, we rescale $\cP$ by a factor of say, $\tfrac 1 \l C_0 k$, to produce a corresponding set of points $\bar \cP\sbs [0,\frac{\l}{C_0 \sqrt k}]^2$ which necessarily admits a TSP path of at most $\leq \l$; note that $\bar \cP$ satisfies the same properties P1-P3 above, with $\e_0$ rescaled to $\tfrac{\l \e_0}{C_0 k}$.

Finally, we modify this configuration (as indicated in Figure \ref{f.pforce}) by adding two points $x,y$ to the set.  With $\bar \cP$ centered at the origin $(0,0)$, we take $x=(-1,-\beta),y=(1,-\beta)$, where $\beta\gg\sqrt \l$ is chosen sufficiently small so that $p,q$ are the closest points on the rhombus $R$ to $x$ and $y$, respectively.  Thus $x$ and $y$ are $\tfrac{\l \e_0}{C_0 k}$ closer to $p$ and $q$ than to any other point in $\bar \cP$.  We call the resulting set $\cQ(\cA,N)\sbs [-1,1]\times [0,1]$.

Essentially, the point set $\cQ$ ensures that any optimal path passing through it will transit the Papadimitriou set optimally, by ensuring the optimal paths will only enter/exit $\cQ$ where we expect. 
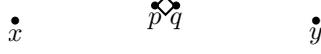
\begin{figure}[t]
\begin{center}
\begin{pdfpic}
\psset{unit=.1cm,dotsize=3pt}
\begin{pspicture}(-5,-5)(5,5)

\rput{45}(0,2){\psframe(-1,-1)(1,1)}
\psdots(-1.414,2)
\psdots(1.414,2)
\rput(-1.414,0){$p$}
\rput(1.414,0){$q$}
}

\psdots(-20,0)(20,0)
\rput(-20,-2){$x$}
\rput(20,-2){$y$}
\end{pspicture}
\end{pdfpic}
\end{center}
\caption{\label{f.pforce} $\cQ$ forces a path through a Papadimitriou path.}
\end{figure}

\begin{lemma}\label{qlemma}
Let $(\cA,N)$ be an instance of the Set Partition problem, and $\cQ=\cQ(\cA,N)$ with $|\cQ|=k$. There is a sufficiently large $D_0$, such that \textbf{If} we have that
\begin{enumerate}[label=(\roman*)]
\item $\cQ\approx_\d Q\sbs [0,t]^d$ for $\d=\tfrac{\l \e_0}{10 C_0 k}$
\item$w,z\in [0,t]^d$ with $\dist(\{w,z\},Q)\geq D_0$, 
\end{enumerate}
\textbf{then} the shortest TSP path $W$ from $w$ to $z$ in $Q\cup \{w,z\}$ has the property that $W$ transits the approximate Papadimitriou set $P\approx_\d \bar \cP$ in $Q$ optimally, from $p$ to $q$.
\end{lemma}
Sets $Q\approx_\d \cQ$ will serve as the sets $S$ to which we apply the definition of scalefreeness.  Very roughly speaking, we will eventually be aiming to contradict $\PP\neq \NP$, since a polynomial-time algorithm to predict optimal paths through $Q$'s would seem to solve the Set Partition problem in polynomial time.
\begin{proof}
Recall that $\cQ$ is constructed by adding two points $x,y$ to the set $\bar \cP$. By construction, the shortest path covering $\cQ$ has endpoints $x,y$, and is of length $<2+\l+2\b$.   Moreover, it is apparent that any path covering $\cQ$ which does not have the endpoint pair $\{x,y\}$ has length at least $3$.  Finally, our choice of $\d$ ensures that the accumulated error in path-lengths when comparing paths in $Q$ vs $\cQ$ is less than $2(k+1)\d<\l$.  Now we suppose that in the shortest path $W$, $w$ is adjacent to $a$ and $z$ is adjacent to $b$, where $\{a,b\}\neq \{x',y'\}\sbs Q$, where $x',y'$ correspond to $x,y\in \cQ$.  Since $W$ is shortest, we must have that
\[
\dist(w,a)+\dist(z,b)+3\leq \dist(w,x')+\dist(z,y')+2+2\l+2\b,
\]
and so
\begin{equation}\label{l.xy}
\dist(w,a)+\dist(z,b)\leq \dist(w,x')+\dist(z,y')-1+2\l+2\b.
\end{equation}
Similarly, we have
\begin{equation}\label{l.yx}
\dist(w,a)+\dist(z,b)\leq \dist(w,y')+\dist(z,x')-1+2\l+2\b.
\end{equation}
So we suppose now that \eqref{l.xy} and \eqref{l.yx} hold simultaneously.  Moreover, let us assume without loss of generality that we have one of the following three cases:\\ 
\textbf{Case 1:}  $\dist(w,x')<\dist(w,y')$ and $\dist(z,y')<\dist(z,x')$, or\\
\noindent\textbf{Case 2:} $\dist(w,x')<\dist(w,y')$ and $\min\{\angle w x'y',\angle wy'x'\}$ is at least $\g>0$, for some $\g$ depending on $\b$, or\\
\noindent \textbf{Case 3:} $wx'$ and $zp$ is the shortest pair of independent edges joining $\{w,z\}$ to $Q$.

Before finishing the proof for each case, let's verify that for $D_0$ large, these cases do indeed cover all cases.  If either $\min\{\angle w x' y',\angle w y' x'\}$ or $\min\{\angle z x'y',\angle z y' x'\}$ are at least $\g$, then we are already in Case 2, by appropriate choices of the labels $w,z,x',y'$ from the available pairs.  If on the other hand both $\min\{\angle w x' y',\angle w y' x'\}$ and $\min\{\angle z x'y',\angle z y' x'\}$ are at most $\g$, then either the angle from $w$ to the center of $Q$ to $z$ is in $(\pi-\g,\pi+\g)$, and we are in Case 1 with the correct choice of which endpoints of $Q$ are called $x',y'$, or the angle is less than $\g$, and we are in Case 3 with a suitable choice of labels.

\noindent \textbf{Case 1:} In this case, making $D_0$ large ensures that $\dist(w,x')-\dist(w,\a)$ and $\dist(z,y')-\dist(z,b)$ are bounded by a number arbitrarily close to 0, violating \eqref{l.xy} (or \eqref{l.yx}, if we had flipped the roles of $w$ and $z$).  So this case cannot occur in simultaneously with \eqref{l.xy} and \eqref{l.yx}.

\noindent \textbf{Case 2:} The second condition of this case implies that for any $\a\in Q$, we have that $\dist(w,x')<\dist(w,a)+\l$ (for sufficiently large $D_0$), which allows us to modify \eqref{l.xy} to the inequality
\beq{98}
\dist(z,y')-\dist(z,b)\geq 1-3\l-2\b.
\eeq
But $z$ is at at least some fixed positive angle from the line through $x,y$ (to which all points in $Q$ are arbitrarily close).  Thus, making $D_0$ large ensures that $\dist(z,a)$ varies by an arbitrarily small amount as we vary $\a\in Q$, contradicting \eqref{98}.  

\noindent \textbf{Case 3:}  Recall that the shortest path on $Q$ goes from $x$ to $p$, then through $\bar \cP$ optimally to $q$, and then ends at $y$.  In particular, this path has endpoint pair $\{x,y\}$.

Let us consider the lengths of shortest paths through $Q$ for choices of endpoints other than $\{x,y\}$.  In particular, we claim that the next-best pair of endpoints is either $\{x,p\}$, and $\{q,y\}$.  For the pair $x,p$, the short path goes from $x$ to $y$ to $q$, through $\bar \cP$ to $p$.  The path for the pair $y,q$ goes from $y$ to $x$ to $p$, through $\bar \cP$ to $q$ (which of these choices gives rise to a shorter path depends on the precise rounding $Q\approx \cQ$).  These paths both have length $3+\beta^2/2+\Theta(\lambda+\beta^3)$ (here we use that $\sqrt{1+\beta^2}\approx 1+\beta^2/2$).  And now we will show that any path with a pair of endpoints other than $\{x,y\}, \{x,p\},$ or $\{q,y\}$ will be longer than any of these choices.  

We prove this as follows.  First let us consider the case where both endpoints lie in $\bar \cP$.  In this case, both $x$ and $y$ are internal vertices of the path through $Q$, and thus it has length at least 4.  Thus we may assume that one endpoint is $x$, while the other is some vertex $v\in \bar \cP$.   We consider two cases: suppose from $v$ the path transits all of $\bar \cP$ before leaving $\bar \cP$; in this case, the optimal choice (over all choices for $v\in \bar \cP$) is clearly to begin from $p$, transit $\bar \cP$ optimally ending in $q$, proceed to $y$, and then return to $x$, as before.   Suppose instead that from $v$, the path visits some of $\bar \cP$, then visits $y$, then returns to visit the rest of $\bar \cP$, before exiting to $\bar \cP$ to $x$.  In this case, the total length used is roughly $3+2\beta^2/2+\Theta(\lambda)$.

Thus, since $wx$ and $zp$ is the shortest pair of independent edges joining $\{w,z\}$ to $Q$, we must have either that the shortest path from $w$ to $z$ through $Q$ uses these edges and and takes the optimum path in $Q$ from $x$ to $p$, or else that it takes a path in $Q$ which is shorter than this optimum path from $x$ to $p$.  But from above, we see that each such choice transits $\bar \cP$ optimally.
\end{proof}

\subsection{The set $\cM_H$} We use $\cQ$ to construct a larger set $\cM_H$.  $\cM_H$ consists of: 
\begin{enumerate} 
\item a copy of $\cQ$, rescaled by a factor $\a$ to lie in $B(0,1)$,  and
\item the set $Y=Y_H\sbs B(0,\sqrt R)$ from the definition of scalefreeness.
\end{enumerate}
The set $\cM_H$ ensures that hypothesis \ref{p.protected} of Definition \ref{d.scalefree} is satisfied at points centering approximations of $\cM_H$.  

\subsection{Perturbing input $\cM_H$ for simulation of $A_H$ by a Turing machine}
The heart of our proof will consist of using the path-finding heuristic $A_H$ to attempt to find a good path through the set $\cM_H$, which would require solving the Papadimitrou set, and contradicting the assumption that $\PP\neq \NP$.  We will show that the failure of $A_H$ to transit $\cM_H$ optimally leads to a significant excess length in the tour $T_H$ found by the heuristic $H$, since a significant number of suitable approximate copies of $\cM_H$ occur throughout the random point set, just by chance.

For the behavior of $A_H$ on $\cM_H$ to reliably predict the behavior of the approximate copies throughout the tour $T_H$, however, we will need to know that the behavior of $A_H$ is stable to small perturbations in the positions of the points in $\cM_H$.  Indeed, this may not be the case in general, but at least, we want to know that we can efficiently, with a deterministic algorithm, perturb the points of $\cM_H$ to produce a set $\cM^g_H$ on which the behavior of $A_H$ is stable with respect to (even smaller) perturbations.

This leads to the following question, which one might hope would be easier than resolving the computational status of comparing sums of radicals:
\begin{question}Given $K$, is there an $L$ and a polynomial-time deterministic algorithm which takes as input a set $X$ of $N$ points (at polynomial precision), and a parameter $\ep=2^{-N^K}$, and outputs a perturbed set of points $X'\approx_\ep X$, such that for all nontrivial choices of $\xi_i\in \{-1,0,1\}$, we have that $\abs{\sum \xi_i x_i}\geq 2^{-N^L}$?
\end{question}
This would imply that we can ``round'' the points of $X$ so that they are not only in general position, but have that property that any small perturbation of them would be equivalent from the standpoint of the signs of sums of the form $\sum \xi_i d_i$.

We do not need to answer to this question here, however.  Instead, we take advantage of the fact that $A_H$ is required to terminate after only polynomially many comparisons, to prove the following Lemma, which offers an easier way out:

\begin{lemma}\label{l.detround}
  For a polynomial path-finding heuristic $A_H$, a set $Z$ of $N$ points, and a constant $K$, we can determinstically, in polynomial time, find a set $Z'\approx_\ep Z$ for $\ep=2^{-N^K}$ such that each of the polynomially many sums $\sum_{i=1}^{\binom N 2} \xi_i d_i$ evaluated by $A_H$ on input $Z'$ will be separated from $0$ by at least $2^{-N^L}$, for a constant $L$ depending only on $K$ and $A_H$.
\end{lemma}

Thus this Lemma accomplishes the same thing as a positive answer to the previous question would, but only for sums $\sum \xi_i d_i$ which will actually be used by $A_H$ on the rounded set.

\begin{proof}
  Recall that the path-finding heuristic is an algorithm which runs in polynomial time in the number of points in its input, with access to a blackbox for making comparisons among sums of distances between input points. The Lemma we are proving claims the existence of a deterministic Turing machine to round the points of the input set $Z$.  To carry out this procedure, we begin by simply running $A_H$ on the input $Z$.  At each time $t=1,2,\dots,\poly(N)$ when a comparison $C=\sum{\xi_i^{(t)} d_i}$ is requested of the blackbox, we do the following:
  \begin{enumerate}
  \item Compute $\sum{\xi_i^{(t)} d_i}$ to precision $2^{-(N^{K}+10t+8)}$.
  \item IF the result lies within $2^{-(N^{K}+10t+4)}$ of 0, THEN:
    \begin{description}
    \item[{\bf Step A}] Change the position of some points by at most (in total) distance $2^{-(N^K+10t+2})$ so that after the change, the same sum differs by at least $2^{-(N^K+10t+4)}$ from 0.
    \end{description}
  \item Accept the resulting status (positive/negative) of the sum as if it had been returned by the blackbox, and continue running $A_H$.
  \end{enumerate}

  (The implementation of Step A is discussed below.) In this way, as $A_H$ runs on $Z$, we perturb the input points repeatedly so that all comparisons requested of the blackbox can be distinguished from 0 by a Turing machine operating at polynomial precision.

Consider some comparison step $t$ after which the $t$th sum $\sum \xi_i^{(t)} d_i$ was guaranteed to differ by at least $2^{-(N^K+10t+4)}$ from 0.  In each subsequent step, this $t$th sum will change slightly, as points of the input are perturbed further.  However, by the triangle inequality, the total impact on this $t$th sum of later perturbations is bounded by
\[
\sum_{\ell=t+1}^\infty 2^{-(N^K+10\ell+2)}=2^{-(N^K+10t+2)}\sum_{\ell'=1}^\infty 2^{-10\ell'}\leq 2^{-(N^K+10t+11)},
\]
which is smaller than the difference $2^{-(N^K+10t+4)}$ guaranteed at step $t$ by the step $t$ perturbation.  As a result, at the end of the this procedure, (after polynomially many perturbations have been carried out, for each of the polynomially many steps of the run of $A_H$), each of the comparisons which $A_H$ requested can be reliably computed on the perturbed point set by a Turing machine simulating $A_H$ at polynomial precision $2^{-N^L}$ for a suitable constant $L$, which depends just on $K$ and the polynomial running time of $A_H$.
\end{proof}

Given $\cM_H$, we let $\cM_H^{g}$ denote the the set resulting from the perturbation procedure of Lemma \ref{l.detround}.
\subsubsection*{Implementing Step A}\label{bb}
We first choose two points $p,q$ such that there exists $i$ such that $\xi_i^{(t)}$ is not zero. Let $d_i=d=|p-q|$. WLOG we let $p=(0,0)$ and $q=(d,0)$.   Let $\r=2^{-(N^K+10t+3)}$ and define the points $p'=p-(\r,0),\,q'=q+(\r,0)$. 

Next let $p_j,j=1,2,\ldots,K$ be the points, other than $q$, such that for some $i$ we have $d_i=|p-p_j|$ and $\xi_i^{(t)}\neq 0$. Let $d_{j}'=|p_j-p'|$ for $i=1,\ldots,K$ and then re-defining $d_j=|p-p_j|$ we let
\[
\D(p)=\left|\sum_{j=1}^K\xi_j^{(t)}(d_j-d_{j}')\right|.
\]
We observe that replacing $p$ by $p'$ results in a new comparison $C_p'$ where
\beq{comp1}
C_p'-C=\D(p)+\z_p\text{ where } \z_p=|p'-q|-|p-q|.
\eeq
Similarly, if we replace $q$ by $q'$ then we obtain a new comparison $C_q'$ where
\beq{comp2}
C_q'-C=\D(q)+\z_q\text{ where }\z_q=|p-q'|-|p-q|.
\eeq
We now consider two cases:\\
{\bf Case 1:} $\max\set{|C_p'|,|C_q'|}\geq \r/2$.\\
In this case the move $p\to p'$ or the move $q\to q'$ implements Step A.

{\bf Case 2:} $\max\set{|C_p'|,|C_q'|}< \r/2$.\\
Suppose now we move $p\to p'$ and $q\to q'$ to obtain a comparison $C_{p,q}'$. Then we have
\beq{comp3}
C_{p,q}'=\D(p)+\D(q)=C_p'+C_q'-(\z_p+\z_q).
\eeq
But,
\[
\z_p=\z_q=\r,
\]
So, from \eqref{comp3}, we see that 
\[
|C_{p,q}'|\geq 2\r-\r=\r,
\]
and so moving $p\to p'$ and $q\to q'$ implements Step A.
\subsection{The set $\Pi_H$} \label{s.angle} We now construct a certain set $\Pi=\Pi(k)$ and show that it constrains the optimal tour in a useful way.  In particular, we let $\Pi(k)$ consist of the four points $\pi_1=(0,5),\: \pi_2=(0,0),\: \pi_3=(1,0),\: \pi_4=(1,5)$ together with all the points $(\frac 1 2, \frac {5j} k)$ for $0\leq j\leq k$ (see Figure \ref{f.Pi}).

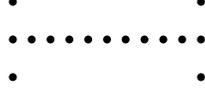
\begin{figure}[t]
\begin{center}
\begin{pdfpic}
\psset{unit=1,dotsize=3pt}
\begin{pspicture}(1,2.5)(0,0)
\psdots(0,0)(2.5,1)(2.5,0)(0,1)
\psdots(0,.5)(.25,.5)(.5,.5)(.75,.5)(1,.5)(1.25,.5)(1.5,.5)(1.75,.5)(2,.5)(2.25,.5)(2.5,.5)
\end{pspicture}
\end{pdfpic}
\end{center}
\caption{\label{f.Pi} The set $\Pi(10)$ (rotated 90 degrees).}
\end{figure}

\begin{lemma}\label{l.Piset}
If $\dist(\{x,y\},\Pi(k))$ is sufficiently large and $P$ is a shortest Hamilton path from $x$ to $y$ in $\Pi(k)\cup \{x,y\}$, then for at least one $i\in \{1,2,3,4\}$ we have that neither neighbor $v_i^1,v_i^2$ of $\pi_i$ on $P$ is in $\{x,y\}$, and moreover that $\dist(\pi_i,v_i^1)$ and $\dist(\pi_i,v_i^2)$ become arbitrarily close to $\frac 1 2$ has $k$ increases; in particular, the neighbors of $\pi_i$ are nearly horizontal translates of $\pi_i$, lying on the line $x=\frac 1 2$.\qed
\end{lemma}

We now define $\Pi_H$ as follows.  We take four copies of the set $\cM_H^{g}$, each rescaled to lie in small balls of radius $\e_\Pi>0$.  In $\Pi(k)$ (for $k=k(H)$ sufficiently large), we replace the four points $\pi_i$ with these copies; those corresponding to $\pi_3$ and $\pi_4$ are reflected horizontally.  (In particular, the resulting set $\Pi_H$ still has horizontal reflection symmetry.)

Suppose we take this $\e_\Pi$ suitably small, that  $U\approx_{\d_H} \Pi_H$ for sufficiently small $\d_H>0$, and that the optimal tour on $X\supseteq U$ transits $U$ in a single path.   Then each approximate copy of $\cM_H^{g}$ in $U$ is transited in single path by the optimal tour (corresponding to hypothesis \ref{p.atapath} from Definition \ref{d.scalefree}) and moreover, Lemma \ref{l.Piset} implies that for at least one of the four copies $M_1,M_2,M_3,M_4$ of $\cM_H^{g}$, $H$ either satisfies hypothesis \ref{p.angle} or else pays an additive error.

\subsection{The set $\Pi_H^3$} We let $\Pi_H^3$ denote 3 copies of $\Pi_H$ centered at the vertices of an equilateral triangle of sidelength $2D_1$, say.   This triple configuration ensures that the optimum tour will transit at least one of the copies of $\cM_H^{g}$ in a a single pass.  Indeed, Observations 2.9 and 2.10 from \cite{TSPconstants} now give the following:
\begin{lemma}
\label{l.P3}
Suppose that $D_2$ is a sufficiently large absolute constant, $(\cA,N)$ is an instance of the Set Partition problem, and $\Pi_H^3=\Pi_H^3(\cA,N)$.
\textbf{If} $\Pi_H^3\approx_{\d_0} Z\sbs X\sbs [0,t]^d$ for ${\d_0}=\tfrac{\alpha\e_0}{10^4 C_0 D_0k}$ and  $\dist(Z,X\stm Z)\geq D_2$, \textbf{then} any TSP tour $T$ on $X$ can either be shortened in $\Pi_H^3$ by some additive constant or otherwise has the property that at least one of the (approximate) copies of $\Pi_H$ in $Z$ is traversed (optimally) by a path by $T$.\qed
\end{lemma}
We emphasize that the constants $\alpha\leq 1\ll D_0\ll D_1\ll D_2$ are absolute, independent of $(\cA,N)$.   

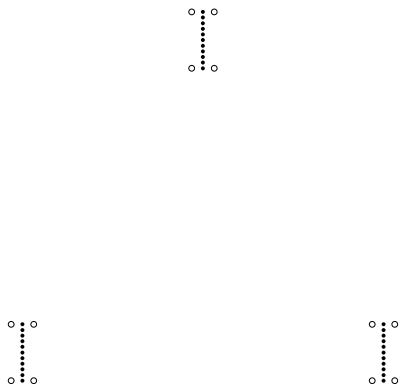
\begin{figure}[t]
\begin{center}
\begin{pdfpic}
\psset{unit=.03cm,dotsize=1.5pt,linewidth=.1pt}
\begin{pspicture}(-150,0)(150,200)
\rput{90}(-80,0){
\pscircle[linewidth=.1pt](0,0){.04cm}
\pscircle[linewidth=.1pt](25,10){.04cm}
\pscircle[linewidth=.1pt](25,0){.04cm}
\pscircle[linewidth=.1pt](0,10){.04cm}
\psdots(0,5)(2.5,5)(5,5)(7.5,5)(10,5)(12.5,5)(15,5)(17.5,5)(20,5)(22.5,5)(25,5)
}

\rput{90}(80,0){
\pscircle[linewidth=.1pt](0,0){.04cm}
\pscircle[linewidth=.1pt](25,10){.04cm}
\pscircle[linewidth=.1pt](25,0){.04cm}
\pscircle[linewidth=.1pt](0,10){.04cm}
\psdots(0,5)(2.5,5)(5,5)(7.5,5)(10,5)(12.5,5)(15,5)(17.5,5)(20,5)(22.5,5)(25,5)
}

\rput{90}(0,138.56){
\pscircle[linewidth=.1pt](0,0){.04cm}
\pscircle[linewidth=.1pt](25,10){.04cm}
\pscircle[linewidth=.1pt](25,0){.04cm}
\pscircle[linewidth=.1pt](0,10){.04cm}
\psdots(0,5)(2.5,5)(5,5)(7.5,5)(10,5)(12.5,5)(15,5)(17.5,5)(20,5)(22.5,5)(25,5)
}



\end{pspicture}
\end{pdfpic}
\end{center}
\caption{\label{f.M} $\Pi_H^3$, consisting of three copies of $\Pi_H$, forces an optimum tour to transit one of 12 copies of $\cM_H^g$ (indicated here as the small circles) in a single pass, from a narrow prescribed angle.}
\end{figure}

\subsection{Using a TSP Heuristic to solve the Set Partition problem}
\label{s.solver}
At this point, in preparation for the proof of Theorem \ref{t.length}, we use the scalefree heuristic $H$ to define the following polynomial time algorithm to solve a Set Partition instance $(\cA,N)$, which will be correct unless Theorem \ref{t.length} holds for $H$.  (Thus, $\PP\neq \NP$ will imply Theorem \ref{t.length}.)

\begin{enumerate}
\item Compute $\cM_H(\cA,N)$ to precision $\frac{\d_0}{4}$.  Using Lemma \ref{l.detround}, perturb the points as necessary to produce a set $M$, for which the path-finding heuristic $A_H$ can be simulated by a Turing machine in polynomial-time.
\item \label{step.enum} Produce a list of paths through $M$ using the algorithm $A_H$.  Because of the rounding $M$ produced by Lemma \ref{l.detround}, comparisons of sums of distances can be done with Turing machine operations.
\item \label{step.L'}Let $L'$ be the minimum length of a path covering one of the three Papadimitriou sets which is a subpath of one of the paths enumerated in step \ref{step.enum}.
\item Return TRUE/FALSE according to whether $L'$ lies within or above the threshold given by \ref{Papa.thres}, respectively.
\end{enumerate}


\subsection*{Proof of Theorem \ref{t.length}}
We suppose that $H$ is scalefree.  $\PP\neq \NP$ implies that there is some instance $(\cA_0,N_0)$ of the Set Partition problem for which the algorithm above returns an incorrect answer.  Observe first that it cannot happen that the algorithm returns TRUE when the correct answer to the Set Partition instance is FALSE: when $(\cA_0,N_0)$ is FALSE, property \ref{Papa.thres} implies that there can be no path through the Papadimitriou sets shorter than the threshold below which the algorithm above would return TRUE.

Thus we are to consider the case that the algorithm above returns FALSE even though the correct answer to $(\cA_0,N_0)$ is TRUE.  In this case, no path enumerated by $A$ transits $M$ in such a way that a Papadimitriou set is traversed optimally.  

We prove Theorem \ref{t.length} by showing that this implies there exists an $\e_H>0$ so that the length of the tour found by $H$ through the random set $\ydn\sbs [0,t]^d$ is w.h.p at least $(1+\e_H)$ times the length of the optimal tour $T$.

An $(\e,R)$-copy of $Z$ in $\ydn\sbs [0,t]^d$ is a set $Z'\sbs \ydn$ such that $Z'\approx_\e Z$ and such that $\dist(Z',\ydn\stm Z')\geq R$.  The following Lemma shows that we find a linear number of $(\ep,R)$ copies of any fixed finite set in the random set $\ydn$ (see Observation 3.1 from \cite{TSPconstants}):
\begin{lemma}\label{l.findit}
Given any finite point set $S$, any $\e,\d>0$, and any $R$, we have that the number $\zeta_S$ of $(\e,R)$-copies $S'$ of $S$ in a random set $X=\ydn\sbs [0,t]^d$, such that the points of $S'$ are $\delta$-distance separated, satisfies 
\beq{zeta}
\zeta_S\geq C_{S,R,\e} n\quad{\text{w.h.p.}}
\eeq
for some constant $C_{S,R,\e}>0$.\qed
\end{lemma}



Now we take $\e_2$ to be the minimum of $\tfrac {\d_0} 8$ and the parameter $\e$ from the definition of scalefreeness for the heuristic $H$, take $R=D_2$, take $S=\Pi_H^3=\Pi_H^3(\cA_0,N_0)$, and use Lemma \ref{l.findit} to find a linear number of $(\tfrac {\e_2} {(Kd+1)^{12}},R+\e_2)$ copies $Z$ of $\Pi_H^3$ which are $\d$-distance separated.

We say that such a 
copy $Z$ of $\Pi_H^3$ has the property $\Lambda_X$ if the tour $T_H$ can be shortened within $Z$ by $\d_1$ for some sufficiently small but fixed $\d_1>0$, 
and we let $\nu_\Lambda$ denote the number of 
copies $Z$ of $\Pi_H^3$ with property $\Lambda_X$.

\noindent\textbf{Claim:} There exists $C>0$ so that if $H$ is scalefree, then $\nu_\Lambda\geq Cn$ w.h.p.

Note that the claim immediately implies the theorem: in the rescaled torus $[0,t]^d$, the heuristic pays a total error of $\d_1\cdot \nu_\Lambda$, and rescaling by $t=n^{1/d}$, this gives Theorem \ref{t.length}.


\begin{proof}[Proof of the Claim]
  Each 
  $\d$-separated copy $Z$ itself consists of three different copies $Z_1,Z_2,Z_3\approx_\ep\Pi_H(\cA_0,N_0)$, and when $Z$ fails to have property $\Lambda_X$, Lemma \ref{l.P3} implies that at least one of the copies $Z_j$ is transited by $T_H$ in a single path.    Fixing a choice of such a copy $Z_j$, Lemma \ref{l.Piset} gives that at least one of the four copies $M_i$ of $\cM^g_H$ in $Z_j$ (and so one of the twelve copies of $\cM^g_H$ in $Z$) satisfies hypotheses \ref{p.atapath} and \ref{p.angle} of Definition \ref{d.scalefree}.  Of course, by construction of $\cM_H$, hypothesis \ref{p.protected} is satisfied at the center of the copy.  Moreover, with at most finitely many exceptions, we may assume that $\cM_H$ satisfies hypothesis \ref{p.exceptional}.  Moreover, by construction of the perturbation $\cM_H^g$ of $\cM_H$ using Lemma \ref{l.detround}, $\cM_H^g$ can be rounded to polynomial precision and still satisfy \ref{p.general}.   Thus all the hypotheses of the implication in Definition \ref{d.scalefree} are satisfied for this copy $M_i$ of $\cM_H^g$.  In particular, we conclude that the tour $T_H$ traverses this copy $M_i$ with one of the paths output by $A_H$ on input $\cM_H^g$; by hypothesis (since our proposed Set Partition algorithm output FALSE on the instance), the path through $M_i$ can be shortened by $\d_1$.
\end{proof}

\section{Branch and Bound}\label{BB}
In our paper \cite{TSPconstants} we considered branch and bound algorithms for solving the Euclidean TSP. Branch-and-bound is a pruning process, which can be used to search for an optimum TSP tour.  Branch-and-bound as we consider here depends on three choices:
\begin{enumerate}
\item A choice of heuristic to find (not always optimal) TSP tours;
\item A choice of lower bound for the TSP;
\item A branching strategy (giving a \emph{branch-and-bound} tree).
\end{enumerate}
For us a branch-and-bound tree is a rooted tree $T_{B\&B}$ where each vertex $v$ is labeled by a 4-tuple $(b_v,\Omega_v,I_v,O_v)$. Here $I_v,O_v$ are disjoint sets of edges and $\Omega_v$ is the set of tours $T$ such that $T\supseteq I_v$ and $T\cap O_v=\emptyset$. The value $b_v$ is some lower bound estimate of the minimum length of a tour in $\Omega_v$ e.g. the optimal value of the Held-Karp linear bound relaxation \cite{DFJ54}, \cite{HK1}, \cite{HK2}. In addition there is an upper bound $B$, which is the length of the shortest currently known tour, found by some associated heuristic. This is updated from time to time as we discover better and better tours. If the root of the tree is denoted by $x$ then we have $I_x=O_x=\emptyset$. 

In \cite{TSPconstants} we allowed essentially any branching strategy. Given $\xdn$, we allowed any method to produce a tree satisfying the following:
\begin{enumerate}[label=(\alph*)]
\item When $v$ is a child of $u$, $I_v\supseteq I_u$ and $O_v\supseteq O_u$.
\item If the children of $u$ are $v_1,\dots,v_k$, then we have $\Omega_{u}=\bigcup_{i=1}^k \Omega_{v_i}.$
\item \label{enum.leaves} The leaves of the (unpruned) branch-and-bound tree satisfy $|\Omega_v|=1$.
\end{enumerate}

This process terminates when the set $L$ of leaves of the pruned branch-and-bound tree satisfies $v\in L\implies b_v\geq B$; such a tree corresponds to a certificate that the best TSP tour found so far by our heuristic is indeed optimum. It is clear that if $v\in L$ then $\Omega_v$ does not contain any tours better than one we already know.

In \cite{TSPconstants} we concentrated on showing that even if we had access to the exact optimum i.e. letting $B=\l$, the minimum length of a tour, none of a selected set of natural lower bounds would be strong enough to make the branch and bound tree polynomial size.   Note that this result does not depend on the branching process itself being efficient.

The aim of this section is to show that even if $b_v=\l$ then a certain branching strategy will fail. Unlike in \cite{TSPconstants}, we cannot allow any branching strategy for our present result, as we might (though extreme computation in the branching process) find that we directly branch to a vertex $w$ where $I_w$ is exactly the set of edges of the shortest tour, giving $B=\l$, causing the algorithm to terminate, given that $b_v=B$ for all leaves of the tree.

It turns that to prove our result, we will need only a mild restriction on the branching strategies allowed.
\begin{enumerate}
\item A vertex $v$ of out tree has two children $w_+,w_-$. Here $I_{w_+}=I_v\cup\set{e},O_{w_+}=O_v$ and $I_{w_-}=I_v,O_{w_+}=O_v\cup\set{e}$  for some edge $e$.
\item The branch and bound tree is explored in a breadth first manner i.e. if the root is at level 0, we do not produce vertices of level $k+1$ until all vertices at level $k$ have been pruned or branched on. 
\end{enumerate}
Note that this captures most branching strategies used in practice, which typically are using an LP-based lower bound on the length of the tour, and branching on the binary values possible for fractionally-valued variables in the linear program.

Theorem \ref{thmbb} will follow easily from the following claim: let $H$ denote some scalefree heuristic and for a vertex $v$ of $\TB$ let $H(v)$ denote the length of the tour constructed by $H$, when it accounts for $I_v,O_v$. Let $\l(v)$ denote the length of the shortest tour in $\Omega_v$.  
\begin{lemma}\label{lemdone}
There exist constants $\e_1,\e_2$ such that w.h.p. if vertex $v$ is at depth at most $k_1=\e_1 n$ then $H(v)\geq \l(v)+\e_2n^{1/2}$.
\end{lemma}
\begin{proof}
Let $\a_0n$ be the minimum number of copies of $Z$ with property $\Lambda_X$ promised by our analysis above, and let $\a_1n^{1/2}$ be a lower bound on the penalty paid by our heuristic for each copy of $Z$. Then if $\e_1=\a_0/2$ we have the lemma for $\e_2=\a_0\a_1/2$.  
\end{proof}
Here we have used the fact that at depth at most $k_1$, w.h.p. there will be a linear number of copies of $Z$ that are unaffected by $I_v,O_v$. These copies provide the necessary increases over the optimum.

It follows from Lemma \ref{lemdone} that for $v$ at depth at most $k_1$ we have
$$H(v)\geq \l(v)+\e_2n^{1/2}\geq B+\e_2n^{1/2}.$$
This means that $v$ is not a leaf. It follows that w.h.p. there will be at least $2^{k_1}=e^{\Omega(n)}$ leaves and Theorem \ref{thmbb} follows.
\proofend
\section{Further work}
From among the heuristics used in practice, the major omission from the present manuscript are the $k$-opt improvement heuristics, and their relatives (such as the Lin-Kernighan heuristic).  Are they scalefree in our sense (or a related sense for which Theorem \ref{t.length} holds)?

\label{s.rand}

\end{document}

\newpage

Let $\th_j=\angle p' pp_j$ for $j=1,2,\ldots,K$. We estimate
\begin{align}
d_{j}'-d_j&=(\r^2+d_j^2-2\r d_j\cos\th_j)^{1/2}-d_j\nonumber\\
&=d_j\brac{\brac{1+\frac{\r^2}{d_j^2}-\frac{2\r\cos\th_j}{d_j}}^{1/2}-1}\nonumber\\
&=-\r\cos\th_j+\frac{\r^2(1-\cos^2\th_j)}{2d_j}+O(\r^3).\label{err}
\end{align}
Now consider the problem
\[
\text{Maximise }\sum_{j=1}^K\frac{x_j^2}{d_j}\text{ subject to }\sum_{j=1}^Kx_k=\a.
\]
If $D=\sum_{j=1}^Kd_j$ then this problem is solved by putting $x_j=\a d_j/D$. In which case, the optimal value for the problem is $\a^2/D$. It follows that if 
\[
\F(p,\th)=\sum_{j=1}^K\cos\th_j=\a
\] 
then
\begin{multline}\label{fff}
\sum_{j=1}^K\brac{-\r\cos\th_j+\frac{\r^2(1-\cos^2\th_j)}{2d_j}}\geq \frac{\r^2}{2}\sum_{j=1}^K\frac{1}{d_j}-\r\a-\frac{\r^2\a^2}{2D}\geq \\
\frac{\r^2K^2}{2D}-\r\a-\frac{\r^2\a^2}{2D}=\frac{\r^2}{2D}\brac{\brac{\a+\frac{D}{\r}}^2 -\brac{K^2+\frac{D^2}{\r^2}}}.
\end{multline}